\documentclass[12pt,a4paper]{article}
\usepackage[latin1,utf8]{inputenc}
\usepackage{amsmath}
\usepackage{amsthm}
\usepackage{amsfonts}
\usepackage{amssymb}
\usepackage{thmtools}
\usepackage{thm-restate}
\usepackage{hyperref}
\usepackage{cleveref}
\usepackage{graphicx}
\usepackage{nicefrac}
\usepackage{fullpage}
\usepackage{bm}
\usepackage{xcolor}
\usepackage[linesnumbered,ruled]{algorithm2e}
\usepackage{tikz}
\usepackage{enumitem}

\newtheorem{thm}{Theorem}[section]
\newtheorem{prop}[thm]{Proposition}

\newtheorem{remark}[thm]{Remark}

\newtheorem{defi}[thm]{Definition}

\newcommand{\Fq}{\mathbb{F}_q}

\newcommand{\cC}{\mathcal{C}}

\newcommand{\bX}{{\bf X}}

\newcommand{\bfv}{{\bf v}}
\newcommand{\bfu}{{\bf u}}

\newcommand\blfootnote[1]{%
  \begingroup
  \renewcommand\thefootnote{}\footnote{#1}%
  \addtocounter{footnote}{-1}%
  \endgroup
}

\newcommand{\F}{\mathbb{F}}
\newcommand{\cR}{\mathcal{R}}

\begin{document}
	
	\title{Optimal Two-Dimensional Reed--Solomon Codes Correcting Insertions and Deletions}
\author{
	Roni Con\thanks{Blavatnik School of Computer Science, Tel Aviv University, Tel Aviv, Israel. Email: roni.con93@gmail.com} \and 
	Amir Shpilka \thanks{Blavatnik School of Computer Science, Tel Aviv University, Tel Aviv, Israel. Email: shpilka@tauex.tau.ac.il. The research leading to these results has received funding from the Israel Science Foundation (grant number 514/20) and from the Len Blavatnik and the Blavatnik Family Foundation.} \and
	Itzhak Tamo \thanks{Department of EE-Systems, Tel Aviv University, Tel Aviv, Israel. Email: zactamo@gmail.com.}
}
\date{}
\maketitle
\begin{abstract}

        Constructing Reed--Solomon (RS) codes that can correct insertions and deletions (insdel errors) has been considered in numerous recent works. Our focus in this paper is on the special case of two-dimensional RS-codes that can correct from $n-3$ insdel errors, the maximal possible number of insdel errors a two-dimensional linear code can recover from.
        It is known (by setting $k=2$ in the lower bound \cite[Proposition 37]{con2023reed}) that an $[n,2]_q$ RS-code that can correct from $n-3$ insdel errors satisfies that  $q=\Omega(n^3)$.
        On the other hand, there are several known constructions of $[n,2]_q$ RS-codes that can correct from $n-3$ insdel errors, where the smallest field size is $q=O(n^4)$. 
        In this short paper, we construct $[n,2]_q$ Reed--Solomon codes that can correct $n-3$ insdel errors with $q=O(n^3)$, thereby resolving the minimum field size needed for such codes. \blfootnote{The work of Itzhak Tamo and Roni Con was supported by the European Research Council (ERC grant number 852953).}

\end{abstract}
\thispagestyle{empty}
\newpage

\section{Introduction}
    

    Constructing error-correcting codes against \emph{synchronization} errors has received a lot of attention recently.   
    The most common model for studying synchronization errors is the insertion-deletion model (insdel for short): an insertion error is when a new symbol is inserted between two symbols of the transmitted word. A deletion is when a symbol is removed from the transmitted word. 
    These are errors that affect the length of the received word. 
    For example, over the binary alphabet, when $100110$ is transmitted, we may receive the word $11011000$, which is obtained from three insertions ($1$ at the beginning and two $0$s at the end) and one deletion (one of the $0$'s at the beginning of the transmitted word). Generally, insdel errors cause the sending and receiving parties to become ``out of sync'' which makes them inherently more difficult to deal with.
	
	Insdel errors appear in diverse settings such as optical recording, semiconductor devices, integrated circuits, and synchronous digital communication networks. 
    Due to the importance of the model and our lack of understanding of some basic problems concerning it, constructing efficient codes that can handle insdel errors is the topic of many recent works  \cite{haeupler2017synchronization,brakensiek2017efficient,guruswami2017deletion,schoeny2017codes,cheng2018deterministic,haeupler2019optimal,cheng2020efficient,guruswami2020optimally,guruswami2021explicit,liu2023explicit} (see also the excellent survey \cite{haeupler2021synchronization}). 
    Further, codes for correcting insdel errors attract a lot of attention due to their possible application in correcting errors in DNA-based storage systems \cite{lenz2019coding,heckel2019characterization,shomorony2022information}. This recent increased interest was paved by substantial progress in synthesis and sequencing technologies. The main advantages of DNA-based storage over classical storage technologies are very high data densities and long-term reliability without an electrical supply.

    Reed-Solomon codes are the most widely used family of codes in theory and practice. Their extensive use can be credited to their simplicity, as well as their efficient encoding and decoding algorithms.    
    Some of their notable applications include QR codes \cite{soon2008qr}, secret sharing schemes \cite{mceliece1981sharing}, and distributed storage systems \cite{tamo2014family,guruswami2016repairing}.
	As such, it is an important problem to understand whether they can also decode from insdel errors. This problem received a lot of attention recently \cite{safavi2002traitor,wang2004deletion,tonien2007construction,duc2019explicit,liu20212,chen2021improved,liu2022bounds,con2023reed}. We note that the best we can hope for when designing a \emph{linear} code that can decode insdel errors is to achieve the \emph{half-Singleton bound}. Specifically, an $[n,k]$ linear code can correct at most $n-2k+1$ insdel errors \cite{cheng2020efficient} (see also \cite{chen2022coordinate,ji2023strict}). In \cite{con2023reed}, the authors showed that there are $[n,k]$ RS codes that achieve the half-Singleton bound where the field size is $n^{O(k)}$.
    This poses the following problem: What is the minimal field size for which there exists an $[n, k]_q$ RS code that can correct $n-2k+1$ insdel errors, the maximum number of insdel errors that can be corrected by a linear code?
    
    In this paper, we provide an answer to this question for RS codes of dimension $k=2$ and show that the minimal field size is $q=\Theta(n^3)$. 
	
\subsection{Basic definitions and notation}
	
	For an integer $k$, we denote $[k]=\{1,2,\ldots,k\}$. 
	Throughout this paper, $\log(x)$ refers to the base-$2$ logarithm. For a prime power $q$, we denote with $\F_q$ the field of size $q$.
	
    	
	A linear code over a field $\F$ is a linear subspace $\cC\subseteq \F^n$. The rate of a linear code  $\cC$ of block length $n$ is $\cR=\dim(\cC)/n$.  Every linear code of dimension $k$ can be described as the image of a linear map, which, abusing notation, we also denote with $\cC$, i.e., $\cC : \F^k \rightarrow \F^n$. 
    When $\cC\subseteq \F_q^n$ has dimension $k$  we say that it is an $[n,k]_q$ code (or an $[n,k]$ code defined over $\Fq$). The minimal distance of a code $\cC$ with respect to a metric $d(\cdot,\cdot)$ is defined as $\min_{\bfv, \bfu \in \cC,\bfu \neq \bfv}{d(\bfv,\bfu)}$. 
	Naturally, we would like the rate to be as large as possible, but there is an inherent tension between the rate of the code and the minimal distance (or the number of errors that a code can decode from).
	In this work, we focus on codes against insertions and deletions. 
	\begin{defi}
		Let $s$ be a string over the alphabet $\Sigma$. The operation in which we remove a symbol from $s$ is called a \emph{deletion} and the operation in which we place a new symbol from $\Sigma$ between two consecutive symbols in $s$, in the beginning, or at the end of $s$, is called an \emph{insertion}. 
		
	\end{defi}

	We next define Reed-Solomon codes (RS-codes from now on). 
	
	\begin{defi}[Reed-Solomon codes]
		Let $\alpha_1, \alpha_2, \ldots, \alpha_n \in\F_q$ be distinct points in a finite field $\mathbb{F}_q$ of order  $q\geq n$. For $k\leq n$ the $[n,k]_q$ RS-code
		defined by the evaluation vector $\bm{\alpha} = ( \alpha_1, \ldots, \alpha_n )$ is the set of codewords 
		\[
		\left \lbrace c_f = \left( f(\alpha_1), \ldots, f(\alpha_n) \right) \mid f\in \mathbb{F}_q[x],\deg f < k \right \rbrace \;.
		\]
	\end{defi}
	
	Namely, a codeword of an $[n,k]_q$ RS-code is the evaluation vector of some polynomial of degree less than $k$ at $n$ predetermined distinct points. 
	It is well known (and easy to see) that the rate of $[n,k]_q$ RS-code is $k/n$ and the minimal distance, with respect to the Hamming metric, is $n-k+1$.

\subsection{Related work} \label{sec:prev-results}
\paragraph{Linear codes against insdel errors.} 
        The basic question of whether there exist good \emph{linear} codes for the insdel model was first addressed in the work of Cheng, Guruswami, Haeupler, and Li \cite{cheng2020efficient}. 
        Specifically, they showed that there are linear codes of rate $R = (1-\delta)/2 - h(\delta)/\log_2(q)$ that can correct from $\delta$ fraction of insdel errors. They also showed an almost matching upper bound which they called the \emph{half-Singleton bound}, given next.
        
	\begin{thm}[Half-Singleton bound: Corollary 5.1 in \cite{cheng2020efficient}]
		Every linear insdel code which is capable of correcting a $\delta$ fraction of deletions has rate at most $(1-\delta)/2 + o(1)$.
	\end{thm}

    \begin{remark}
        \label{main-remark}
        The following  non-asymptotic version of the Half-Singleton bound can be derived from the proof of Corollary 5.1 in \cite{cheng2020efficient}:  
        An $[n,k]$ linear code can correct at most  $n-2k+1$ insdel errors. 
    \end{remark}
        
        Cheng et al. \cite{cheng2020efficient} also constructed the first asymptotically good binary linear codes for insdel errors. Their codes have rate $R< 2^{-80}$ and can correct efficiently from $\delta < 1/400$ insdel errors. Then, Con et al. \cite{con2022explicit} constructed linear codes with better rate-distance tradeoffs, and Cheng et al. \cite{cheng2023linear} constructed asymptotically good codes in the high-noise and high-rate regimes, covering regimes of rate-distance that are not achievable by the codes of \cite{con2022explicit}. We note that the open question of constructing codes over \emph{small alphabets} that achieve the half-Singleton bound is still open. 

\paragraph{RS-codes against insdel errors.}
	To the best of our knowledge,  Safavi-Naini and Wang \cite{safavi2002traitor} were the first to study the performance of RS-codes against insdel errors. They gave an algebraic condition that is sufficient for an RS-code to correct from insdel errors, yet they did not provide any construction. 
 
	Wang,  McAven, and  Safavi-Naini  \cite{wang2004deletion} constructed a $[5,2]$ RS-code capable of correcting a single deletion. Then, in \cite{tonien2007construction}, Tonien and  Safavi-Naini constructed an $[n,k]$ generalized-RS-codes capable of correcting from $\log_{k+1} n - 1$ insdel errors. Con, Shpilka, and Tamo \cite{con2023reed} showed the existence of $[n,k]_q$ RS-codes that can correct from $n-2k+1$ where $q=n^{O(k)}$. These codes are the first linear codes that achieve the half-Singleton bound. They also provided deterministic construction of such code over a field of size $n^{k^{O(k)}}$.
	
 Much attention was given to the specific case of $2$-dimensional RS-codes. Specifically, the goal is to construct $[n,2]_q$ RS-codes correcting $n-3$ insdel errors with the smallest $q$ possible. 
 By Remark \ref{main-remark}, it should be noted that these codes are optimal for correcting insdel errors, as they correct the maximum possible number of such errors. 
 Numerous constructions were published \cite{duc2019explicit,chen2021improved,con2023reed,liu2022bounds} (see table \Cref{tab:prev-results-table}). The current best construction is due to \cite{con2023reed} that presented an $[n,2]_q$ RS-code correcting $n-3$ insdel errors where $q=O(n^4)$. They also proved that the field size of an $[n,2]_q$ RS-code correcting $n-3$ deletions must be $q=\Omega (n^3)$. This work closes the gap and shows that $q=\Theta(n^3)$.  

\begin{table}
    \begin{center}
    \begin{tabular}{ |c|c| } 
     \hline
     Work & Field size  \\ 
     \hline
     \cite{duc2019explicit} & $\exp(n)$ \\
     \hline
     \cite{liu2022bounds} & $O(n^5)$ \\
     \hline
     \cite{con2023reed} & $O(n^4)$ \\
     \hline
      This work &  $O(n^3)$ \\ 
     \hline
    \end{tabular}
    \end{center}
    \caption{Previous $[n,2]_q$ RS-code cosntructions that can correct from $n-3$ insdel errors and their field size. We note that the construction presented in \cite{con2023reed} works only for fields of characteristic $3$ whereas in this work, we provide constructions for every positive characteristics.}
    \label{tab:prev-results-table}
\end{table}
\subsection{Our results}
    Our main result is the following theorem.
	\begin{restatable}{thm}{rsTwoDimConst}
		\label{thm:rs-twodim-const}
		For any $n\geq 3$, there exists an explicit $[n,2]_{q}$ RS-code that can correct from $n - 3$  insdel errors, where $q = O(n^3)$. 
	\end{restatable}
	
    As described earlier, an explicit construction of an RS-code amounts to specifying an evaluation vector. In our constructions, we define sets of size $n$ and prove that any ordering of these $n$ points into an evaluation vector defines a two-dimensional RS-code that can correct from $n-3$. 
    \begin{remark}
        The specific field size obtained in \Cref{thm:rs-twodim-const} is $q = (n+1)^3$. 
        The lower bound, given in \cite[Proposition 37]{con2023reed}, states that when $k=2$, then, $q = \Omega (n^3)$. However, by performing a bit more careful analysis, we get that the lower bound on the field size of an $[n,2]$ RS code correcting $n-3$ is $q \geq \binom{n}{3} - 1$. Therefore, the upper bound obtained in this paper, is tight up to a small constant factor.
    \end{remark}

    \begin{remark}
        We note that there is a similarity between our constructions and those of \emph{$3$-order} Maximum Distance Separable (MDS) codes over small fields presented in \cite[Section 5]{brakensiek2022improved}. Specifically, the evaluation points used in \Cref{sec:second-const} are almost identical to those in their construction. 
        However, the algebraic condition ensuring that these points are good for correcting deletions  differs from the one used in their proof. 
      In particular, while their proof relies on computer assistance to compute a Groebner basis, we can obtain our result directly without such assistance.
    \end{remark}

\subsection{An algebraic condition}	
	In this section, we recall the algebraic condition presented in \cite{con2023reed}. 
	We first make the following definitions:
		We say that a vector of indices $I\in [n]^s$ is an \emph{increasing} vector if its coordinates are monotonically increasing, i.e., for any  $1\leq i<j\leq s$, $I_i<I_j$, where $I_i$ is the $i$th coordinate of $I$.  
	 For two vectors $I,J\in [n]^{2k-1}$  with distinct coordinates we define the following (variant of a) vandermonde matrix of order $(2k-1)\times (2k-1)$ in the formal variables $\bX=(X_1,\ldots,X_n)$:

    \begin{equation} \label{eq:mat-lcs-eq}
	V_{I,J}(\bX)=\begin{pmatrix} 
	1 & X_{I_1} & \ldots & X_{I_1}^{k-1}  & X_{J_1} &\ldots & X_{J_1}^{k-1} \\ 
	1 & X_{I_2} & \ldots & X_{I_2}^{k-1}  & X_{J_2} &\ldots & X_{J_2}^{k-1} \\
	\vdots &\vdots & \ldots &\vdots &\vdots &\ldots &\vdots \\
	1 & X_{I_{2k-1}} & \ldots & X_{I_{2k-1}}^{k-1}  & X_{J_{2k-1}} &\ldots & X_{J_{2k-1}}^{k-1}\\
	\end{pmatrix} .
	\end{equation}

 \begin{prop} \cite[Proposition 2.1]{con2023reed} \label{prop:cond-for-RS}
		Consider the $[n,k]_q$ RS-code defined by an evaluation vector   $\bm{\alpha}=(\alpha_1,\ldots,\alpha_n)$.  
		If for every two increasing vectors $I,J\in [n]^{2k-1}$ that agree on at most $k-1$ coordinates, it holds that $\det(V_{I,J}(\bm{\alpha})) \neq 0$, then the code can correct any $n-2k+1$ insdel errors.
		Moreover, if the code can correct  any $n-2k+1$ insdel errors, then the only possible vectors in $\text{Kernel}\left(V_{I,J}(\bm{\alpha})\right)$ are of the form $(0,f_1,\ldots,f_{k-1},-f_1,\ldots,-f_{k-1})$.
	\end{prop}

\section{Explicit construction for $k=2$ with cubic field size}

	In this section, we prove  \Cref{thm:rs-twodim-const}, which is restated for convenience.
	\rsTwoDimConst*
    The proof of \Cref{thm:rs-twodim-const} will follow from two code constructions of an $[n,2]$ RS-code that can correct from $n-3$ insdel errors, over a field of size $O(n^3)$. The first code construction works for any characteristic, whereas the second construction works only for fields of characteristic not equal to $2$. However, for these characteristics, given the same field size, the latter construction provides a slightly longer code compared to the first construction.

    \subsection{Construction for any characteristic} 
    In this section, we shall present the first construction that works for any finite field. Both the construction and its proof are given in the following proposition. 
    
	\label{RS-k-2-char-2}

    \begin{prop} \label{prop:k-2-explicit-char-2}
    Let $\mathcal{A} \subseteq \mathbb{F}_{q}^*$ be a set of size $n$ such that for any two distinct elements $\delta, \delta' \in \mathcal{A}$,  $\delta \neq  -\delta'$, and let $\gamma$ be a root of a degree $3$ irreducible polynomial over $\Fq$.
    Let the vector $\bm{\alpha} = (\alpha_1, \alpha_2, \ldots, \alpha_{n})$ be some ordering of the $n$ elements $\delta + \delta^{-1} \cdot \gamma, \delta \in \mathcal{A}$. Then, the  $\left[n, 2\right]$ RS-code defined over $\mathbb{F}_{q^3}$ with the evaluation vector $\bm{\alpha}$ can correct any $n-3$ insdel errors.
    

	\end{prop}

    \begin{remark}
    \label{remarkkk}
        Note that if the characteristic of the field $\Fq$ does not equal to $2$, then, necessarily, the length of the code, $n$, is at most $ (q-1)/2$.  As for each $\delta\in \mathcal{A}, -\delta\notin \mathcal{A}.$
        On the other hand, if the characteristic equals to $2$, then it can be as large as $ q-1$.  
        However, in both cases, since the code is defined over $\mathbb{F}_{q^3}$, we have that the field size is $O(n^3)$.
    %
    \end{remark}

    \begin{proof}
        Assume towards a contradiction that the claim is false, then by \Cref{prop:cond-for-RS}   there  exist two vectors of distinct evaluation points $(\beta_1,\beta_2,\beta_3),(\beta_4,\beta_5,\beta_6)$, that agree on at most one coordinate, such that,
		\[
		\left|
		\begin{pmatrix}
		1 &\beta_1 &\beta_4 \\ 
		1 &\beta_2 &\beta_5 \\
		1 &\beta_3 &\beta_6 \\
		\end{pmatrix} 
		\right|
		= 0 
		\;.
		\]
		Equivalently, 
		  \[
                (\beta_1 - \beta_2)(\beta_5 - \beta_6) - (\beta_4 - \beta_5)(\beta_2 - \beta_3) = 0 \;.
            \] 
            Write $\beta_i := \delta_i + \delta_i^{-1} \cdot \gamma$ for $i\in[6]$ and observe that the LHS is a polynomial in $\gamma$ of degree less than $3$ over $\Fq$. Namely, 
            \[
            p_0(\bm{ \delta}) + p_1(\bm{ \delta})\cdot \gamma +p_2(\bm{ \delta})\cdot \gamma^2 = 0
            \]
            where,
        \begin{align*}
               p_0(\bm{ \delta}) &= (\delta_1 - \delta_2)(\delta_5 - \delta_6) - (\delta_2 - \delta_3) (\delta_4 - \delta_5)\\
               p_1(\bm{ \delta}) &= (\delta_1 - \delta_2)(\delta_5^{-1} - \delta_6^{-1}) + (\delta_1^{-1} - \delta_2^{-1}) (\delta_5 - \delta_6)\\
               &\quad- (\delta_2^{-1} - \delta_3^{-1})(\delta_4-\delta_5) - (\delta_2 -\delta_3)(\delta_4^{-1}-\delta_5^{-1}) \\
               p_2(\bm{ \delta}) &= (\delta_1^{-1} - \delta_2^{-1})(\delta_5^{-1} - \delta_6^{-1}) - (\delta_2^{-1} - \delta_3^{-1}) (\delta_4^{-1}-\delta_5^{-1}) \;.
        \end{align*}
        $p_1(\bm{\delta})$ and $p_2(\bm{\delta})$ can be simplified further,
            \begin{align*}
               p_1(\bm{ \delta}) &= (\delta_1 - \delta_2)(\delta_5 - \delta_6) \left( -(\delta_1\delta_2)^{-1} -(\delta_5\delta_6)^{-1}\right) \\
               &\quad - (\delta_2 - \delta_3)(\delta_4 - \delta_5) \left( -(\delta_2 \delta_3)^{-1} - (\delta_4 \delta_5)^{-1} \right) \\
               p_2(\bm{ \delta}) &= (\delta_1\delta_2\delta_5\delta_6)^{-1}(\delta_1 - \delta_2)(\delta_5 - \delta_6) \\ &\quad - (\delta_2\delta_3 \delta_4\delta_5)^{-1}(\delta_2 - \delta_3) (\delta_4 - \delta_5)\;. 
            \end{align*}
            
            Next, since the minimal polynomial of $\gamma$ over $\Fq$ is of degree $3$,  $p_i(\bm{\delta}) = 0$ for  $i=0,1,2$.
         $p_0(\bm{\delta}) = 0$ implies that, 
            \begin{equation} \label{eq:char-2-p-0}
                (\delta_1 - \delta_2)(\delta_5 - \delta_6) = (\delta_2 - \delta_3) (\delta_4 - \delta_5)\neq 0\;,
            \end{equation}
            
            where the inequality follows since the coordinates of each of the vectors $(\delta_1, \delta_2, \delta_3)$ and $(\delta_4,\delta_5, \delta_6)$ are distinct. Substituting \eqref{eq:char-2-p-0}  in the equation $p_2(\bm{\delta}) = 0$ gives,
            \begin{equation} \label{eq:char-2-p-2}
                \delta_1\delta_6 = \delta_3 \delta_4,
            \end{equation}
            and by $p_1(\bm{\delta}) = 0$ and \eqref{eq:char-2-p-2}, 
            \begin{equation} \label{eq:char-2-p-1}
                \delta_1\delta_2 + \delta_5\delta_6 = \delta_2\delta_3 + \delta_4\delta_5 \;.
            \end{equation}
             \eqref{eq:char-2-p-0} and \eqref{eq:char-2-p-2} imply that,
            \[
                \delta_1 \delta_5 + \delta_2 \delta_6 = \delta_2 \delta_4 + \delta_3\delta_5 \;.
            \]
            Subtracting the last  equation from \eqref{eq:char-2-p-1} we get,
            \begin{equation}
            \label{stam-eq}
            (\delta_1 - \delta_6)(\delta_2 - \delta_5) = (\delta_3 - \delta_4) (\delta_2 - \delta_5) \;.
            \end{equation}
            
            If $\delta_2 = \delta_5$ then by 
            \eqref{eq:char-2-p-1} $\delta_1 + \delta_6 = \delta_3 + \delta_4$. Together with  \eqref{eq:char-2-p-2}, we get that both $\{\delta_1,\delta_6\}$ and $\{\delta_3,\delta_4\}$ are solutions to the  quadratic equation,
            \[
            x^2 - (\delta_3 +\delta_4)x + \delta_3 \delta_4 = 0\;.
            \]
            Therefore, 
            $\{\delta_1,\delta_6\}=\{\delta_3,\delta_4\}$, and we arrive at a contradiction to the assumption that two evaluation vectors have distinct coordinates and they agree on at most one coordinate (indeed, if $\delta_1 = \delta_4$ and $\delta_3 = \delta_6$ then $\beta_1 = \beta_4$ and $\beta_3 = \beta_6$).  
Otherwise, $\delta_2 \neq \delta_5$ and \eqref{stam-eq} becomes
            $\delta_1 - \delta_6 = \delta_3 - \delta_4$. Combining it with \eqref{eq:char-2-p-2}, 
            we get that, as before, both $\{\delta_1,-\delta_6\}$ and $\{\delta_3,-\delta_4\}$ are solutions to the  quadratic equation
            $
            x^2 - (\delta_3 -\delta_4)x - \delta_3 \delta_4 = 0
            $,
            and therefore $\{\delta_1,-\delta_6\}=\{\delta_3,-\delta_4\}$. Recall that for $\delta\in \mathcal{A}, -\delta \notin \mathcal{A}$, thus, the only possibility is that $\delta_1 = \delta_3$ and $\delta_4 = \delta_6$ and again we arrived at a contradiction. The result follows. 
    \end{proof}

    \subsection{Improved construction for odd characteristics}
    \label{sec:second-const}
    In this section, we give the second code construction that improves the length of the code when the characteristic of the finite field $\Fq$ does not equal to  $2$ (See Remark \ref{remarkkk}). Specifically, the code's length can be as large as $q-1$. The evaluation points of the constructed RS-code and the proof are very similar to the first construction. 
%
        
	\begin{prop} \label{prop:k-2-explicit}

  	Let $\mathbb{F}_{q}$ be a finite field of characteristic $p>2$ and let $\mathcal{A}\subseteq \Fq^*$ be a subset of size $n$. Let $\gamma$ be a root of a degree $3$ irreducible polynomial over $\Fq$, and let the vector 
    $\bm{\alpha} = (\alpha_1, \alpha_2, \ldots, \alpha_n)$
    be some ordering  of the $n$ elements $\delta + \delta^2 \cdot \gamma, \delta \in \mathcal{A}$. Then, the $\left[n,2\right]$ RS-code defined over $\mathbb{F}_{q^3}$ with the evaluation vector $\bm{\alpha}$ can correct any $n-3$ insdel errors. 
	\end{prop}
        \begin{proof}
        As before, assume towards a contradiction that the claim is false, then by  \Cref{prop:cond-for-RS} there  exist two vectors of distinct evaluation points $(\beta_1,\beta_2,\beta_3)$ and $(\beta_4,\beta_5,\beta_6)$, that agree on at most one coordinate, such that 
		\[
		\left|
		\begin{pmatrix}
		1 &\beta_1 &\beta_4 \\ 
		1 &\beta_2 &\beta_5 \\
		1 &\beta_3 &\beta_6 \\
		\end{pmatrix} 
		\right|
		= 0 
		\;.
		\]
		Equivalently, 
		  \[
                (\beta_1 - \beta_2)(\beta_5 - \beta_6) - (\beta_2 - \beta_3) (\beta_4 - \beta_5) = 0 \;.
            \] 
            Write $\beta_i := \delta_i + \delta_i^2 \cdot \gamma$ and observe that the LHS is a polynomial in $\gamma$ of degree less than $3$ over $\Fq$. Namely, 
            $
            p_0(\bm{\delta}) + p_1(\bm{ \delta})\cdot \gamma +p_2(\bm{ \delta})\cdot \gamma^2 = 0
            $,
            where,
            \begin{align*}
               p_0(\bm{ \delta}) &= (\delta_1 - \delta_2)(\delta_5 - \delta_6) -  (\delta_2 - \delta_3) (\delta_4-\delta_5) \\
               p_1(\bm{ \delta}) &= (\delta_1 - \delta_2)(\delta_5^2 - \delta_6^2) + (\delta_1^2 - \delta_2^2) (\delta_5 - \delta_6)\\
               &\quad- (\delta_2^2 - \delta_3^2)(\delta_4-\delta_5) - (\delta_2 - \delta_3)(\delta_4^2-\delta_5^2) \\
               p_2(\bm{ \delta}) &= (\delta_1^2 - \delta_2^2)(\delta_5^2 - \delta_6^2) - (\delta_2^2 - \delta_3^2) (\delta_4^2-\delta_5^2) \;.
            \end{align*}
            
            Next, by the definition of $\gamma$,  $p_i(\bm{\delta}) = 0$ for  $i=0,1,2$.
         $p_0(\bm{\delta}) = 0$ implies that, 
            \begin{equation} \label{eq:p-0}
                (\delta_1 - \delta_2)(\delta_5 - \delta_6) = (\delta_2 - \delta_3) (\delta_4-\delta_5)\neq 0, 
            \end{equation}
            where the inequality follows since the coordinates of each of the vectors $(\delta_1, \delta_2, \delta_3)$ and $(\delta_4,\delta_5, \delta_6)$ are distinct.
            Substituting \eqref{eq:p-0}  in $p_2(\bm{\delta}) = 0$ and $p_1(\bm{\delta})=0$ gives,
            \begin{equation} \label{eq:p-2}
                (\delta_1 + \delta_2)(\delta_5 + \delta_6) = (\delta_2 + \delta_3)(\delta_4 + \delta_5)\;,
            \end{equation}
            and 
            \begin{equation} \label{eq:restriction-1}
            \delta_1 + \delta_6 = \delta_3 + \delta_4 \;,   
            \end{equation} 
            respectively. 
            Summing \eqref{eq:p-0} and \eqref{eq:p-2} and  
            as the characteristic is different than $2$ we get that,
            \[
            \delta_5(\delta_1 - \delta_3) = \delta_2(\delta_4 - \delta_6) \;,
            \]
            which with  \eqref{eq:restriction-1} implies  that $\delta_2 = \delta_5$. Plugging it in \eqref{eq:p-2}, and taking \eqref{eq:restriction-1} into consideration,  we get that, 
            \begin{equation} \label{eq:restriction-2}
            \delta_1 \delta_6 = \delta_3 \delta_4 \;.
            \end{equation}
            
            By \eqref{eq:restriction-1} and \eqref{eq:restriction-2}, we get that
            $\{\delta_1,\delta_6\}$ and $\{\delta_3,\delta_4\}$ are the solutions to the   quadratic equation
            $
            x^2 -(\delta_3 + \delta_4)x + \delta_3 \delta_4 = 0
            $, and therefore $\{\delta_1,\delta_6\}=\{\delta_3,\delta_4\}$, which as before leads to a contradiction. 
    \end{proof}


    \begin{remark}
        We note that the natural generalization of this construction and proof technique for $k>2$ translates into solving a polynomial system of equations with $4k-2$ variables where the degrees of the resulting polynomials can be as high as $2k(k-1)$. It is possible that a careful analysis of these equations for small values of $k$ will lead to improved constructions as compared to the one given in \cite{con2023reed}. However, in order to handle arbitrary values of $k$ we believe that a different approach is required  due to the complexity of the polynomial equations and lack of apparent structure.
        

    
    
    \end{remark}
	\bibliographystyle{alpha}
	\bibliography{RSinsdel}

\end{document}